\documentclass[aps,pra,nofootinbib,twocolumn,floatfix,showpacs,superscriptaddress]{revtex4-1}
\usepackage[utf8]{inputenc}
\usepackage{amsfonts,amsmath,graphicx,enumerate,epsfig,amsthm,subfigure,newunicodechar,textcomp,color,bm}
\usepackage{amsmath}
\usepackage{pifont,bbding,amssymb,soul,mathrsfs}
\usepackage{indentfirst}
\usepackage{times}
\usepackage{appendix, comment} 
\usepackage{CJKutf8}
\usepackage{ucs}
\usepackage[encapsulated]{CJK}
\usepackage{tikz,pgfplots}
\usepackage{braket}
\usepackage{gensymb}
\usepackage{hyperref}
\hypersetup{colorlinks=true, linkcolor=blue, citecolor=blue, urlcolor=black }
\usepackage[capitalize]{cleveref}
\makeatletter

\newcommand{\Rmnum}[1]{\expandafter\@slowromancap\romannumeral #1@}

\usepackage{hyperref}
\hypersetup{
    colorlinks=true,
    linkcolor=blue,
    citecolor=blue,
    urlcolor=blue,
}

\makeatother

\newcommand{\tr}{{\rm Tr}}
\newtheorem{theorem}{Theorem}
\newtheorem{lemma}{Lemma}

\newcommand{\appref}[1]{\hyperref[#1]{{Appendix~\ref*{#1}}}}
\newcommand{\be}{\begin{eqnarray} \begin{aligned}}
\newcommand{\ee}{\end{aligned} \end{eqnarray} }
\newcommand{\benn}{\begin{eqnarray*} \begin{aligned}}
\newcommand{\eenn}{\end{aligned} \end{eqnarray*}}

\begin{document}
%\title{Experimental Quantification of Entanglement and Coherence without Tomography }
\title{Quantification of Entanglement and Coherence with Purity Detection}
\author{Ting Zhang}
\affiliation{School of Physics, State Key Laboratory of Crystal Materials, Shandong University, Jinan 250100, China}
\author{Graeme Smith}
\affiliation{JILA, University of Colorado/NIST, 440 UCB, Boulder, CO 80309, USA}
\affiliation{Center for Theory of Quantum Matter, University of Colorado, Boulder, Colorado 80309, USA}
\affiliation{Department of Physics, University of Colorado, 390 UCB, Boulder, CO 80309, USA}
\author{John A. Smolin}
\affiliation{IBM T.J. Watson Research Center, 1101  Kitchawan Road, Yorktown Heights, NY 10598}
\author{Lu Liu}
\affiliation{School of Physics, State Key Laboratory of Crystal Materials, Shandong University, Jinan 250100, China}
\author{Xu-Jie Peng}
\affiliation{School of Physics, State Key Laboratory of Crystal Materials, Shandong University, Jinan 250100, China}
\author{Qi Zhao}
\affiliation{QICI Quantum Information and Computation Initiative, Department of Computer Science,
The University of Hong Kong, Pokfulam Road, Hong Kong}
\author{Davide Girolami}
\affiliation{DISAT, Politecnico di Torino, Corso Duca degli Abruzzi 24, Torino 10129, Italy}
\author{Xiongfeng Ma}
\affiliation{Center for Quantum Information, Institute for Interdisciplinary Information Sciences, Tsinghua University, Beijing, 100084 China}
\author{Xiao Yuan}
\email{xiaoyuan@pku.edu.cn}
\affiliation{Center on Frontiers of Computing Studies, Peking University, Beijing 100871, China}
\affiliation{School of Computer Science, Peking University, Beijing 100871, China}

\author{He Lu}
\email{luhe@sdu.edu.cn}
\affiliation{School of Physics, State Key Laboratory of Crystal Materials, Shandong University, Jinan 250100, China}
\affiliation{Shenzhen Research Institute of Shandong University, Shenzhen 518057, China}
%\date{\today}

\begin{abstract}
Entanglement and coherence are fundamental properties of quantum systems, promising to power near future quantum technologies, such as quantum computation, quantum communication and quantum metrology. Yet, their quantification, rather than mere detection, generally requires reconstructing the spectrum of quantum states, i.e., experimentally challenging measurement sets that increase exponentially with the system size. Here, we demonstrate quantitative bounds to operationally useful entanglement and coherence that are universally valid, analytically computable, and experimentally friendly. Specifically, our main theoretical results are lower and upper bounds to the coherent information and the relative entropy of coherence in terms of local and global purities of quantum states. To validate our proposal, we experimentally implement two purity detection methods in an optical system: shadow estimation with random measurements and collective measurements on pairs of state copies. The experiment shows that both the coherent information and the relative entropy of coherence of pure and mixed unknown quantum states can be bounded by purity functions. Our research offers an efficient means of verifying large-scale quantum information processing. 
\end{abstract}

\maketitle
\section*{Introduction}
Entanglement is a fundamental trait of many-body quantum systems and a key resource for quantum information processing~\cite{nielsen2011quantum, dense,cryptography91, Teleporting93, Horodecki09}. Recently, theoretical methods to characterize quantum superpositions have been generalized to evaluate quantum coherence in single systems~\cite{Colloquium17} and explore its uses for quantum technologies~\cite{giovannetti2011advances, Demonstrating19, Catalytic14,lostaglio2015description,narasimhachar2015low,romero2014quantum,huelga2013vibrations,lloyd2011quantum,lambert2013quantum}. Quantification of such resources provides insights on the true computational power of quantum devices~\cite{vidal2002computable,hayden2001asymptotic, entdet,BDSW,wootters}, and many important measures are defined in terms of the von Neumann entropy $S(\rho)=-\tr(\rho\log\rho)$. Besides, the von Neumann entropy has been found widespread applications in quantum date compression~\cite{Schumacher95coding}, quantum thermodynamics~\cite{Brand13Thermal}, capacity bounds for quantum channels~\cite{Junaid2022PRA} and many-body physics, from the characterization of topological matter~\cite{kitaev06Topological,Furukawa07dimer,satzinger2021realizing}, to dynamics out of equilibrium~\cite{alba2013entanglement}, to the understanding of tensor network methods~\cite{gerster2016superfluid}~(see Ref.~\cite{laflorencie2016quantum} for a review). However, the quantification of von Neumann entropy is hard both theoretically and experimentally, as it necessitates knowledge of the full spectrum of the system state $\rho$. Clever methods that enable to \emph{witness} entanglement and coherence employ randomized measurements~\cite{brydges2019probing, PhysRevA.99.052323,elbenrandomized,elben2020mixed} and collective detections on many copies of quantum states to extract spectrum polynomials, e.g., the state purity $\tr(\rho^2)$ ~\cite{Massar95,Tarrach99,islam2015measuring, Bagan06,wu2021experimental,zhangDetecting2017}. Yet, these protocols cannot be easily applied to {\it quantify} entanglement and coherence: there are not measures of quantum resources which can be expressed in terms of directly observable (polynomial) quantities.

In this letter, we address this challenge by proposing an efficient approach to identify quantitative bounds to entanglement and coherence of unknown quantum states in terms of purity functions, in contrast to other protocols based on local measurements~\cite{Mandal2020PRR,huang2022measuring}. We focus on the \emph{coherent information} and \emph{relative entropy of coherence}, which are both defined in terms of the von Neumann entropy and are information measures with compelling operational interpretations. The coherent information is related to the distillable entanglement and the capacity of quantum channels with applications in  quantum communication, one-way entanglement distillation, quantum state merging and quantum many-body physics~\cite{Schumacher96, Capacity97,schum, Lloyd, Devetak,horodecki2005partial,uncertainty21,friis2019entanglement,brydges2019probing}. The relative entropy of coherence lower bounds the distillable coherence and plays an important role in quantum thermodynamics, quantum metrology, quantum computing, quantum random number generation and quantum phase transitions~\cite{ResourceCoherence16, zhao2019one,yang,cohrev}. We prove analytical upper and lower bounds on the coherent information and relative entropy of coherence of arbitrary finite dimensional quantum states in terms of their local and global purities, which are measurable without spectrum reconstruction~\cite{multipartiteding}. Then, we experimentally demonstrate our proposal in an optical system by implementing the randomized measurements scheme on four-qubit states, and collective measurements on two copies of two-qubit states. The experiment results confirm that operationally useful entanglement and coherence of unknown quantum states can be quantified without spectrum reconstruction. 

\section*{Results}
Our study has two main merits. First, it discovers simple analytical functions that {\it quantify}, rather than only {\it witness}, key quantum resources in arbitrary systems of finite dimension. Second, it shows an experimental comparison between the well-established interference-based method for non-tomographic exploration of quantum properties~\cite{REN2017281, Cincio_2018,hou2018deterministic,wu2019experimentally}, and the recently introduced ``shadow estimation" techniques~\cite{huang2020predicting,18M120275X, Shadows21,DaleyGrowth12}. Together, our study provides a theoretically universal and practically efficient means to benchmark  features of unknown quantum systems.

\subsection{Quantification of coherent information.}
For quantum states $\rho_{AB}\in \mathcal{H}_{d_A}\otimes \mathcal{H}_{d_B}$, the coherent information is defined by 
\begin{equation}
I(A\rangle B)=S(\rho_B)-S(\rho_{AB}),
\end{equation}
where $A$ and $B$ are subsystems and $\rho_B=\tr_{A}(\rho_{AB})$ is the reduced density matrix on subsystem $B$. A positive value of $I(A\rangle B)$ signals operationally useful entanglement between subsystems $A$ and $B$~\cite{uncertainty21}. 

Measuring $I(A\rangle B)$ requires knowledge of the eigenvalues of the density matrices. We propose a method to obtain upper and lower bounds on the von Neumann entropy in terms of the global and marginal purity of the state. Given the spectral decomposition of a $d$-dimensional quantum state, i.e., $\rho = \sum_{i=1}^d\lambda_{i,\rho}\ket{\psi_i}\bra{\psi_i}, \sum_i \lambda_{i,\rho}=1, \langle\psi_i|\psi_j\rangle=\delta_{ij},  \lambda_{1,\rho}\ge \lambda_{2,\rho} \ge \dots \lambda_{d,\rho}$,  we determine the extreme values of the state entropy $S(\rho)= -\sum_{i=1}^d\lambda_{i,\rho}\log \lambda_{i,\rho}$ at fixed purity $ \mathcal P(\rho):=\sum_{i=1}^d\lambda_{i,\rho}^2$, where the logarithm is written in base 2~\cite{zyczkowski2003renyi}.
The spectrum $\{\lambda_{i,\rho}^{\text{M}}\}$ that maximizes $S(\rho)$ is 
 $\lambda_{1,\rho}^{\text{M}} = \frac{1}{d}+\sqrt{\frac{d-1}{d}\left(\mathcal P(\rho)-\frac{1}{d}\right)},
  \lambda_{2,\rho}^{\text{M}} =\dots=\lambda_{d,\rho}^{\text{M}} = \frac{1-\lambda_{1,\rho}^{\text{M}}}{d-1}.$
The spectrum $\{\lambda_{i,\rho}^{\text{m}}\}$ that minimizes $S(\rho)$ is given by
$
\lambda_{1,\rho}^{\text{m}}=\lambda^{\text{m}}_{2,\rho}=\dots=\lambda^{\text{m}}_{k_\rho-1,\rho}=\frac{1-\alpha_\rho}{k_\rho-1},
\lambda^{\text{m}}_{k_\rho,\rho} = \alpha_\rho,
\lambda^{\text{m}}_{k_\rho+1,\rho}=\dots=\lambda^{\text{m}}_{d,\rho} = 0
$,
where $\alpha_\rho={1}/{k_\rho} - \sqrt{ (1-1/k_\rho)(\mathcal P(\rho)-1/k_\rho)}$ and $k_\rho$ is the integer such that $ \frac{1}{k_\rho} \le  \mathcal P(\rho) < \frac{1}{k_\rho-1}$. We can immediately use these results to bound the coherent information as follows (see Supplementary Note 1 for details).

\textbf{Result 1}---Given a quantum state $\rho_{AB}$, its coherent information $I(A\rangle B)$ is bounded as follows:
 \begin{equation} \label{icthebound}
      l_e(\rho_{AB})\leq I(A\rangle B)\leq u_e(\rho_{AB})   
 \end{equation}
 where
 \begin{widetext}
\begin{equation}\label{Eq:CIbounds}
\begin{split}
&l_e(\rho_{AB})=(\lambda^{\text{m}}_{k_{\rho_B},\rho_B}-1) \log \lambda^{\text{m}}_{1,\rho_B}
 -\lambda^{\text{m}}_{k_{\rho_B},\rho_B} \log  \lambda^{\text{m}}_{k_{\rho_B},\rho_B}+ (1-\lambda^{\text{M}}_{1,\rho_{AB}}) \log \frac{(1-\lambda^{\text{M}}_{1,\rho_{AB}})}{(d-1)}+ \lambda^{\text{M}}_{1,\rho_{AB}} \log {\lambda^{\text{M}}_{1,\rho_{AB}}},\\
&u_e(\rho_{AB})=(1-\lambda^{\text{m}}_{k_{\rho_{AB}},\rho_{AB}}) \log \lambda^{\text{m}}_{1,\rho_{AB}}+\lambda^{\text{m}}_{k_{\rho_{AB}},\rho_{AB}} \log  \lambda^{\text{m}}_{k_{\rho_{AB}},\rho_{AB}}- (1-\lambda^{\text{M}}_{1,\rho_B}) \log \frac{(1-\lambda^{\text{M}}_{1,\rho_B})}{(d_B-1)} - \lambda^{\text{M}}_{1,\rho_B} \log \lambda^{\text{M}}_{1,\rho_B}.\\
\end{split}
\end{equation}
 \end{widetext}
The lower and upper bounds is tight for pure states~(${\cal P}(\rho_{AB})=1$) with $\mathcal{P}(\rho_B)=\frac{1}{d_B}$ and the difference $\epsilon_e=\mathcal{P}(\rho_B)-1/d_B$ certifies the tightness of $u_e(\rho)$ and $l_e(\rho)$.

\subsection{Quantification of quantum coherence.} 
In a way similar to how non-factorizable superpositions of multipartite states, e.g. $\sum_i c_i \ket{ii\ldots i }$, yield entanglement, the quantumness of a system can be identified with the degree of coherence of its state $\ket{\psi}=\sum_i c_i \ket{i}, \sum_i |c_i|^2=1,$ in a reference basis $\{\ket{i}\}$. 
One natural way to quantify the coherence of a state in a reference basis $\{\ket{1},\ket{2},\dots,\ket{d}\}$ of a $d$-dimensional Hilbert space $\mathcal{H}_d$ is by measuring how far it is to the set of incoherent states ${\cal I}$~\cite{BCP,herbut}. The choice of distance function is in principle arbitrary. Yet, an important operational interpretation is enjoyed by the relative entropy of coherence~\cite{BCP}
\begin{equation}\label{CREEQ}
  C_{\mathrm{RE}}(\rho) = \min\limits_{\sigma\in{\cal I}}S(\rho||\sigma)=S(\rho_{d}) - S(\rho),
\end{equation}
where $\rho_{d}=\sum_i \ket{i}\!\!\bra{i}\rho \ket{i}\!\!\bra{i}$ is the state after dephasing in the reference basis. In the asymptotic limit of infinite system preparations, $C_{\mathrm{RE}}(\rho)$ represents the maximal rate of extraction of maximally coherent qubit states $1/2\sum_{i,j=0,1}\ket{i}\bra{j}$ from $\rho$ by incoherent operations. Like the coherent information, this quantity is bounded by purity function (see~Supplementary Note 1 for details).

\begin{figure*}[ht!]
\centering
\includegraphics[width=\linewidth]{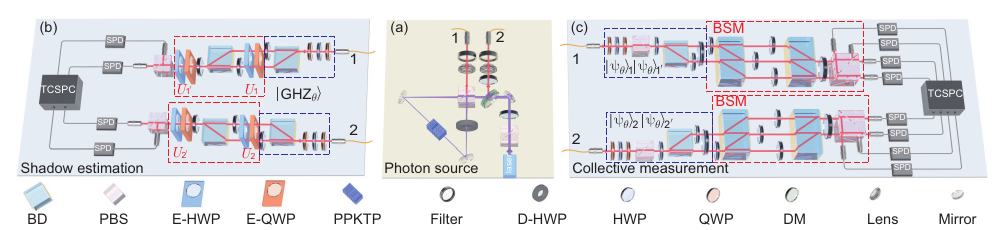}
\caption{\textbf{Schematic illustration of the experimental setup.} $\textbf{a}$ Generation of the biased polarization-entangled state $\cos\theta\ket{HH}_{12}+\sin\theta\ket{VV}_{12}$. $\textbf{b}$ Setup to extend $\cos\theta\ket{HH}_{12}+\sin\theta\ket{VV}_{12}$ into $\ket{\text{GHZ}_{\theta}}=\cos\theta\ket{HhHh}_{11^\prime22^\prime}+\sin\theta\ket{VvVv}_{11^\prime22^\prime}$, and demonstrate the shadow estimation scheme. $\textbf{c}$ Setup to prepare two-copy states and implement the collective measurement scheme. Symbols used in $\textbf{a}$, $\textbf{b}$, and $\textbf{c}$~BD: beam displacer; PBS: polarization beam splitter; SPD: single-photon detector; DM: dichroic mirror; E-HWP: electrically-rotated HWP; E-QWP: electrically-rotated QWP; D-HWP: dual-wavelength HWP; TCSPC: time-correlated
single-photon counting system. The abbreviation BSM represents Bell-state measurement.\label{Fig:setup}}
\end{figure*} 

\textbf{Result 2}---The relative entropy of coherence $C_{\mathrm{RE}}(\rho)$ is bounded as follows:
\begin{equation}\label{Eq:Estimationlb}
  l_c(\rho)\leq C_{\mathrm{RE}}(\rho)\leq u_c(\rho),
\end{equation}
where
\begin{widetext}
\begin{equation}\label{Eq:boundsofCR}
\begin{split}
&l_c(\rho)= (\lambda^{\text{m}}_{k_{\rho_d},\rho_d}-1) \log \lambda^{\text{m}}_{1,\rho_d}-\lambda^{\text{m}}_{k_{\rho_d},\rho_d} \log \lambda^{\text{m}}_{k_{\rho_d},\rho_d}
  + (1-\lambda^{\text{M}}_{1,\rho})\log\frac{(1-\lambda^{\text{M}}_{1,\rho})}{(d-1)} + \lambda^{\text{M}}_{1,\rho} \log{\lambda^{\text{M}}_{1,\rho}},\\
&u_c(\rho)=  (1-\lambda^{\text{m}}_{k_{\rho},\rho}) \log \lambda^{\text{m}}_{1,\rho} +\lambda^{\text{m}}_{k_{\rho},\rho} \log \lambda^{\text{m}}_{k_{\rho},\rho}
 -  (1-\lambda^{\text{M}}_{1,\rho_d}) \log \frac{(1-\lambda^{\text{M}}_{1,\rho_d})}{(d-1)} - \lambda^{\text{M}}_{1,\rho_d} \log \lambda^{\text{M}}_{1,\rho_d}.\\
\end{split}
\end{equation}
\end{widetext}
This inequality chain, like the one in Eq.~\ref{icthebound}, is tight for pure states~(${\cal P}(\rho)=1$) with diagonal matrix of $\rho_d=\frac{1}{d}\mathbb I_d$~($\mathcal{P}(\rho_d)=\frac{1}{d}$). The difference $\epsilon_c=\mathcal{P}(\rho_d)-1/d$ certifies the tightness of $u_c(\rho)$ and $l_c(\rho)$. $\epsilon_e(\epsilon_c)\to0$ indicates the maximally entangled state~(maximally coherent state), which is of particular interest in quantum information science.

\begin{figure*}[ht!]
\centering
\includegraphics[width=\linewidth]{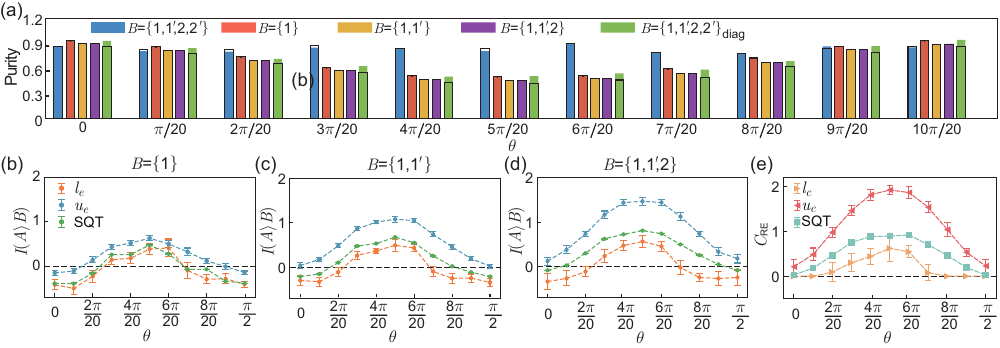}
\caption{\textbf{Experimental results of quantification of $I(A\rangle B)$ and $C_{\mathrm{RE}}$ on the prepared $\rho_{\text{GHZ}_{\theta}}$ by shadow estimation.} $\textbf{a}$ The estimation of global purity $\hat{\mathcal P}(\rho_{\text{GHZ}_{\theta}})$, marginal purity $\hat{\mathcal P}(\rho_B)$ and the purity of diagonal matrix $\hat{\mathcal P}(\rho_{\text{GHZ}_{\theta, d}})$. The colored bars represent the results from shadow estimation, while the black frames represent the results from SQT for comparison. $\textbf{b}-\textbf{d}$ The upper bound $u_e$ and lower bound $l_e$ of $I(A\rangle B)$ with $B=\{1\}$, $B=\{1, 1^\prime\}$ and $B=\{1, 1^\prime, 2\}$ respectively. $\textbf{e}$ The upper bound $u_c$ and lower bound $l_c$ of $C_{\mathrm{RE}}(\rho_{\text{GHZ}_{\theta}})$. The error bars represent the statistical error by repeating shadow estimation for 10 times.\label{Fig:CI}}
\end{figure*}

\subsection{Detecting purity with shadow estimation.}
We first use shadow tomography~\cite{huang2020predicting, PRXQuantumRobust21, PRXQuantumStruchalin21} to detect the purity of the four-qubit biased Greenberger-Horne-Zeilinger (GHZ) states in the form of
\begin{equation}\label{Eq:GHZ}
\ket{\text{GHZ}_{\theta}}=\cos\theta\ket{HhHh}_{11^\prime22^\prime}+\sin\theta\ket{VvVv}_{11^\prime22^\prime},
\end{equation}
which are encoded on the polarization and path degrees of freedom~(DOF) of photons. As shown in~Fig.~\ref{Fig:setup}$\textbf{a}$, the polarization-entangled photons are generated from a periodically poled potassium titanyl phosphate~(PPKTP) crystal set at Sagnac interferometer. Then, we then sent two photons into two beam displacers~(BDs) as shown in~Fig.~\ref{Fig:setup}$\textbf{b}$, which transmits the vertical polarization and deviates the horizontal polarization. Consequently, the biased GHZ state $\ket{\text{GHZ}_{\theta}}$ is obtained, where $h$ ($v$) denotes the deviated~(transmitted) spatial mode.

We prepare eleven $\rho_{\text{GHZ}_{\theta}}$ by setting $\theta\in[0, \frac{\pi}{2}]$ with interval of $\frac{\pi}{20}$, and then use $M=2\times10^4$ measurements in shadow estimation on each $\rho_{\text{GHZ}_\theta}$ to bound the coherent information $I(A\rangle B)$ of $\rho_{\text{GHZ}_{\theta}}$. We consider the bipartition of $\rho_{\text{GHZ}_{\theta}}$ with two subsystems $A$ and $B$, where $A\cup B=\left\{1, 1^\prime, 2, 2^\prime\right\}$ and $A\cap B=\varnothing$. Each subsystem contains $|A|$ and $|B|$ qubits, respectively. We consider three cases of $B=\{1\}$, $B=\{1, 1^\prime\}$ and $B=\{1, 1^\prime, 2\}$. The unbiased estimator of purities $\mathcal P_{\rho_{\text{GHZ}_{\theta}}}$ and $\mathcal P_{\rho_B}$ are constructed with $\{\hat{\rho}_{\text{GHZ}_{\theta}}^{(m)}\}$ by~\cite{huang2020predicting}
\begin{equation}\label{Eq:purity_CS}
\hat{\mathcal P}(\rho_{\text{GHZ}_{\theta}})=\frac{1}{M(M-1)}\sum_{m\neq m^{\prime}}\mathrm{Tr}\left[\hat{\rho}_{\text{GHZ}_{\theta}}^{(m)}\hat{\rho}_{\text{GHZ}_{\theta}}^{(m^{\prime})}\right]
\end{equation}
and
\begin{equation}\label{Eq:purity_sub}
 \hat{\mathcal P}(\rho_B)=\frac{1}{M(M-1)}\sum_{m\neq m^{\prime}}\mathrm{Tr}\left[\hat{\rho}_{B}^{(m)}\hat{\rho}_{B}^{(m^{\prime})}\right], 
\end{equation}
where $ \hat{\rho}_B=\bigotimes_{n\in B}3U_n^\dagger\ket{b_n}\bra{b_n}U_n-\mathbb{I}_2$. The results of $\hat{\mathcal P}(\rho_{\text{GHZ}_{\theta}})$ and $ \hat{\mathcal P}(\rho_B)$ are shown in~Fig.~\ref{Fig:CI}$\textbf{a}$. To indicate the accuracy of estimated purities, we perform standard quantum tomography~(SQT)~\cite{vogel1989determination,leonhardt1995quantum,white1999nonmaximally} on the prepared $\rho_{\text{GHZ}_{\theta}}$ with $1.4\times10^6$ measurements, and treat the reconstructed state as target state. With the reconstructed $\rho_{\text{GHZ}_{\theta}}$, we calculate the corresponding purities that are shown with black frames in~Fig.~\ref{Fig:CI}$\textbf{a}$. The maximal error between purities~(Eq.~\ref{Eq:purity_CS} and~Eq.~\ref{Eq:purity_sub}) estimated from classical shadows and SQT is $\epsilon=0.0132\pm0.0109$. The high accuracy~($\epsilon\ll1$) agrees well with the theoretical prediction that the measurement cost of shadow tomography is in the order of $2^{|AB|}/\epsilon^2$~\cite{elben2020mixed}, while the SQT requires~(at least) an order of $2^{|AB|}\text{rank}(\rho_{AB})/\epsilon^2$ measurements to reach the same accuracy~\cite{Tomography17,Efficientqt}. According to Eq.~\ref{Eq:CIbounds}, the lower bound $l_e$ and upper bound $u_e$ of $I(A\rangle B)$ can be calculated with the estimated purities,  and the results are shown with orange and blue dots in~Fig.~\ref{Fig:CI}$\textbf{b}$-Fig~\ref{Fig:CI}$\textbf{d}$ respectively. We observe that $l_e>0$ with $\theta=\frac{3\pi}{20}, \frac{4\pi}{20}, \frac{5\pi}{20}$ and $\frac{6\pi}{20}$, which indicates the corresponding $\rho_{\text{GHZ}_{\theta}}$ admits distillable entanglement. To investigate the tightness of lower and upper bounds of $I(A\rangle B)$, we calculate the $I(A\rangle B)$ with reconstructed $\rho_{\text{GHZ}_{\theta}}$ instead of theoretical predictions as $I(A\rangle B)$ is sensitive to noise~(See Supplementary Note 2 for analyzations). The results of calculated $I(A\rangle B)$) are shown with green dots in~Fig.~\ref{Fig:CI}$\textbf{b}$-Fig.~\ref{Fig:CI}$\textbf{d}$, in which we observe that $I(A\rangle B)$) is well bounded by $l_e$ and $u_e$ expect $\theta=6\pi/20$ in Fig.~\ref{Fig:CI}$\textbf{b}$. Similar phenomena are also observed in Fig.~\ref{Fig:CI}$\textbf{a}$, where the estimation of $\hat{\mathcal P}(\rho_{\text{GHZ}_{\theta, d}})$~(green bars) are larger than the results from SQT.  There are two main reasons attributed to these discrepancies. The first one is that the randomized measurement and SQT are performed separately, i.e., they are not obtained from the same copies of prepared $\rho_{\text{GHZ}_\theta}$. There are unavoidable noises such as the slight drifts of the mounts holding BDs, which would accordingly introduces errors in state preparation and detection. The second one is that we use maximal likelihood estimation~(MLE) in SQT to return a physical state from collected data . MLE is a biased estimation which underestimates properties of unknown quantum state~\cite{Schwemmer2015PRL}, while the shadow tomography we implemented is an unbiased estimation of purity~\cite{huang2020predicting}.

\begin{figure*}[ht!]
\centering
\includegraphics[width=\linewidth]{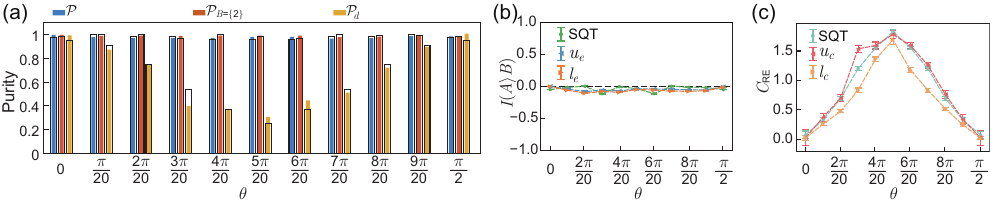}
\caption{\textbf{Experimental results of quantification of $I(A\rangle B)$ and $C_{\mathrm{RE}}$ of $\rho_{\psi_{2,\theta}}$ by collective measurements.} $\textbf{a}$ The estimated purities of $\mathcal P(\rho_{\psi_{2,\theta}})$, $\mathcal P(\rho_{\psi_{2,\theta,B}})$ and $\mathcal P(\rho_{\psi_{2,\theta,d}})$. $\textbf{b}$ The upper bound $u_e$ and lower bound $l_e$ of $I(A\rangle B)$ with $B=\{2\}$. $\textbf{c}$ The upper bound $u_c$ and lower bound $l_c$ of $C_{\mathrm{RE}}(\rho_{\psi_{2,\theta}})$. The error bars represent standard deviations obtained from conducting the experiment ten times.\label{Fig:RE}}
\end{figure*}

To bound $C_{\mathrm{RE}}(\rho_{\text{GHZ}_{\theta}})$, we calculate the purity of the diagonal matrix of $\rho_{\text{GHZ}}$ by $\hat{\mathcal P}(\rho_{\text{GHZ}_{\theta, d}})=\sum_{i=1}^{16}{d_{i}^{2}}$ with $d_i$ being the diagonal elements of $\hat{\rho}_{\text{GHZ}_{\theta}}=\sum_{m=1}^M\hat{\rho}_{\text{GHZ}_{\theta}}^{(m)}$. The results of $\hat{\mathcal P}(\rho_{\text{GHZ}_{\theta, d}})$ are shown with green bars in~Fig.~\ref{Fig:CI}$\textbf{a}$. Thus, $u_c$ and $l_c$ are deduced with estimated $\hat{\mathcal P}(\rho_{\text{GHZ}_{\theta, d}})$ and $\hat{\mathcal P}(\rho_{\text{GHZ}_{\theta}})$ according to~Eq.~\ref{Eq:boundsofCR}. As $C_{\mathrm{RE}}\geq 0$, we set $l_c=0$ whenever it takes negative values. The results of the calculated $u_c$ and $l_c$ are shown with red and yellow triangles in~Fig.~\ref{Fig:CI}$\textbf{e}$, in which one observes they tightly bound $C_{\text{RE}}(\rho_{\text{GHZ}_\theta})$ from SQT (cyan squares).

\subsection{Detecting purity with collective measurements.} 
The purity of a quantum state $\rho$ can be indicated from two copies of $\rho$ by $\mathcal P(\rho)=\tr(\rho^2)=\tr(\mathbb V\rho \otimes \rho)$ with $\mathbb V$ being the swap operation on $\rho\otimes\rho$~\cite{Bovino2005PRL,yuan2020direct,roik2022entanglement,conlon2023approaching}. The purity from collective measurement has been demonstrated to extract Renyi entropy for violation of entropic inequalities to witness entanglement~\cite{Bovino2005PRL}. This Renyi quantity, while able to certify entanglement as it is an entanglement witness~\cite{terhal}, does not quantify it. We consider the case of two-qubit state in the form of $\ket{\psi_{2, \theta}}=\ket{\psi_\theta}_1\ket{\psi_\theta}_2$ with $\ket{\psi_\theta}_1=\ket{\psi_\theta}_2=\cos\theta\ket{0}+\sin\theta\ket{1}$. Experimentally, $\ket{\psi_{2, \theta}}$ is encoded in the polarization DOF and the setup to generate $\ket{\psi_{2, \theta}}$ as shown in~Fig.~\ref{Fig:setup}$\textbf{c}$. We first post-select the component $\ket{H}_1\ket{H}_2$ using two polarizing beam splitters~(PBSs). By applying a HWP that transforms $\ket{H}$ to $\cos\theta\ket{H}+\sin\theta\ket{V}$ individually on photon 1 and photon 2, $\ket{\psi_{2, \theta}}$ is obtained. The copy of $\ket{\psi_{2, \theta}}$ is encoded in the path DOF, i.e., $\ket{\psi_\theta}_{1^\prime}=\ket{\psi_\theta}_{2^\prime}=\cos\theta\ket{h}+\sin\theta\ket{v}$.

The swap operation on $\mathbb V$ on $\rho\otimes\rho$ can be implemented by performing Bell-state measurement~(BSM) between each qubit and its corresponding copy~\cite{PhysRevLett.93.110501,islam2015measuring,PhysRevLett.113.170401}. In our case, the BSM is performed between the polarization-encoded qubit $1(2)$ and the path-encoded qubit $1^\prime(2^\prime)$~\cite{Activating21}, respectively. The outcome probability of the two BSMs on $\rho_{12}\otimes\rho_{1^\prime2^\prime}$ is denoted by $p_{ij}= \mathrm{Tr}[(\Pi_i \otimes \Pi_j)\rho_{\psi_{2,\theta}}\otimes\rho_{\psi_{2,\theta}}]$, where 
$\Pi_1=\ket{\Psi^+}\bra{\Psi^+}$, $\Pi_2=\ket{\Psi^-}\bra{\Psi^-}$, $\Pi_3=\ket{\Phi^+}\bra{\Phi^+}$, and $\Pi_4=\ket{\Phi^-}\bra{\Phi^-}$
are projectors onto Bell states $\ket{\Psi^{\pm}}=(\ket{Hv}\pm \ket{Vh})/\sqrt{2}$ and  $\ket{\Phi^{\pm}}=(\ket{Hh}\pm \ket{Vv})/\sqrt{2}$. The purity of $\rho_{\psi_{2,\theta}}$ and the subsystem purity of $\rho_{\psi_{2,\theta,B}}$ with $B=\{2\}$  are then obtained by
\begin{equation}\label{BSM2}
\begin{aligned}
\mathcal P(\rho_{\psi_{2,\theta}})=1-2(p_{12}+p_{32}+p_{42}+p_{21}+p_{23}+p_{24}),
\end{aligned}
\end{equation}
and
\begin{equation}\label{BSMD}
\begin{aligned}
\mathcal P(\rho_{\psi_{2,\theta,B}})=1-2(p_{12}+p_{22}+p_{32}+p_{42}).
\end{aligned}
\end{equation}

Similarly, the purity of the diagonal matrix of $\rho_{\psi_{2,\theta}}$ can be obtained by
\begin{equation}\label{BSMd}
\begin{aligned}
\mathcal P(\rho_{\psi_{2,\theta,d}})=1-2p_{33}+2p_{44}-p_{11}. 
\end{aligned}
\end{equation}  
The results of $\mathcal P(\rho_{\psi_{2,\theta}})$, $\mathcal P(\rho_{\psi_{2,\theta, B}})$ and $\mathcal P(\rho_{\psi_{2,\theta,d}})$ are shown in~Fig.~\ref{Fig:RE}$\textbf{a}$, with $\theta\in[0, \frac{\pi}{2}]$ with interval of $\frac{\pi}{20}$. The lower bound $l_e$ and upper bound $u_e$ of $I(A\rangle B)$ are calculated according to~Eq.~\ref{Eq:CIbounds} and shown in~Fig.~\ref{Fig:RE}$\textbf{b}$. We observe $u_e<0$ for all $\rho_{\psi_{2,\theta}}$, which indicates the prepared $\rho_{\psi_{2,\theta}}$ is less useful for entanglement distillation. Similarly, the lower bound $l_c$ and upper bound $u_c$ of $C_{\mathrm{RE}}(\ket{\psi_{2, \theta}})$ can be calculated according to~Eq.~\ref{Eq:boundsofCR}. The results are shown in~Fig.~\ref{Fig:RE}$\textbf{c}$. Note that $l_c$ is much closer to $u_c$ compared to the case in~Fig.~\ref{Fig:CI}$\textbf{e}$. This is because the bounds $l_c$ and $u_c$ are functions of the leading order term (purity) in Taylor expansion of the von Neumann entropy about pure states, so that $l_c$ and $u_c$ are tight for pure states. Experimentally, the prepared $\rho_{12}$ and $\rho_{1^\prime2^\prime}$ are quite close to the ideal form of  $\ket{\psi_{2, \theta}}$, while $\rho_{\text{GHZ}_{\theta}}$ is much more noisy. The high accuracy of $l_c$ and $u_c$ is also confirmed by $C_{\mathrm{RE}}(\rho_{12})$ with reconstructed $\rho_{12}$ from SQT, which is shown with cyan dots in~Fig.~\ref{Fig:RE}$\textbf{c}$. 

\section*{Discussion}
We demonstrated universal and computable theoretical bounds to operationally meaningful measures of entanglement and coherence in terms of purity functionals.
Then, we experimentally extracted these bounds by implementing two purity detection methods: shadow estimation and collective measurements. The experiment showed that quantum resources can be {\it estimated}, rather than just {\it witnessed}, with a precision that does not scale with the rank of the state (guaranteed by theory~\cite{Efficientqt,Tomography17,huang2020predicting,elben2020mixed}), conversely to state tomography. The scalability of the measurement network makes purity detection employable in testing the successful preparation of quantum superpositions in large computational registers, certifying that a complex device has run a truly quantum computation. The proposed bounds are sufficiently tight for practically useful quantum states, i.e., the high-fidelity GHZ-like states or maximally coherent states, which are important entanglement and coherence resources that are widely used in quantum information protocols. The bounds Eqs.~(\ref{icthebound}) and (\ref{Eq:Estimationlb}) represent the leading order term in Taylor's expansion of the von Neumann entropy. Thus, tightened bounds for noisy states can be extracted by evaluating the higher-order terms
$\mathrm{Tr}(\rho^3), \mathrm{Tr}(\rho^4), \ldots,
\mathrm{Tr}(\rho^{d})$, which can be efficiently detected with hybrid shadow estimation~\cite{zhou2022hybrid,peng2024experimental}. In particular, the bounds become strict when we include moments of the system dimension. It would be interesting for future work to study the tightness of the bounds for the intermediate cases. Another unexplored direction is that one can extend the method proposed here to determine directly measurable bounds to the total correlations in multipartite systems $\{A_i\}$. For instance, consider the quantum analogue of the multi-information between random variables~\cite{han,modi}
\begin{equation}{\cal I}(\rho_{A_1,\ldots,A_n}) =\min\limits_{\bigotimes_i\sigma_{A_i}}S\left (\rho_{A_1,\ldots,A_n}||\bigotimes_i\sigma_{A_i}\right).
\end{equation}
It is easy to verify that the product of the state marginals $\bigotimes_i\rho_{A_i}$ solves the minimization, ${\cal I}(\rho_{A_1,\ldots,A_n})=\sum_i S(\rho_{A_i})-S(\rho_{A_1,\ldots,A_n})$. Quantitative bounds to the total system correlations in terms of purities are given by a straightforward generalization of~Eq.~\ref{icthebound}. 

Our work has important and wide practical applications in various fields in quantum computation, communication, quantum thermodynamics, quantum many-body physics, etc. The proposed method has an immediate application in benchmarking current and near-term quantum technologies and serves as a basic and useful tool for analyzing and optimizing practical implementations of quantum information protocols.

\section*{Methods}
\subsection{Biased polarization-entangled photon source.}
We use a continuous-wave laser operating at a central wavelength of 405 nm with a full width at half maximum (FWHM) of 0.012~nm as our pump light source. The pump light passes through a PBS followed by an HWP set at $\theta/2$, which transforms the polarization of the pump light into $\cos\theta\ket{H}_{p}+\sin\theta\ket{V}_{p}$. The pump light passes PBS that transmits the component of $\ket{H}$ and reflects component of $\ket{V}$. Then, the PPKTP crystal is coherently pumped from anticlockwise and clockwise directions respectively, and the generated photons are superposed on the PBS leading to the outcome state of $\cos\theta\ket{HV}_{12}+\sin\theta\ket{VH}_{12}$. An HWP set at $45^{\circ}$ is applied on photon 2, which leads to a biased polarization-entangled state in form of $\cos\theta\ket{HH}_{12}+\sin\theta\ket{VV}_{12}$. To enhance collective efficiency, we employ lens L1 with a focal length of 200~mm and lens L2 with a focal length of 250~mm. The two photons pass through narrowband filters~(NBFs) with an FWHM of 3~nm and then are coupled into single-mode fibres.

\subsection{Shadow tomography.}  In shadow tomography, local random unitary operations $U_n\in\text{Cl}_2$ are individually applied on each qubit of an $N$-qubit state $\rho$, where $\text{Cl}_2$ is the single-qubit Clifford group. Then the rotated state is measured on the Pauli-$Z$ basis, producing a bit string $\ket{b}=\ket{b_1b_2\cdots b_N}, b_n\in\{0, 1\}$. The classical shadow of a single experimental run is constructed by $\hat{\rho}=\bigotimes_{n=1}^N 3U_n^\dagger\ket{b_n}\bra{b_n}U_n-\mathbb{I}_2$ with $\mathbb{I}_2$ being identity matrix. By repeating the measurement $M$ times, one has a collection of classical shadows $\{\hat{\rho}^{(m)}\}$ which is further exploited for the estimation of various properties of the underlying state $\rho$~\cite{elben2020mixed, brydges2019probing}.  The random unitary operations $U_n\in\text{Cl}_2$ on the polarization and path DOF are implemented with a combination of electrical-controlled half waveplate~(E-HWP) and quarter waveplate~(E-QWP)~\cite{Shadows21}, and the projective measurements on the Pauli-$Z$ basis are sequentially performed on the polarization and path DOF~(See Supplementary Note 2 for more details).

\section*{Data Availability}
The data that support the findings of this study have been deposited in the Zenodo
database with the identifier https://zenodo.org/records/11386676

\section*{Code availability}
The code supporting the findings of this study have been deposited in the Zenodo
database with the identifier https://zenodo.org/records/11386676

\section*{Acknowledgements}
The authors thank the anonymous reviewers for the insightful comments on the work. This work is supported by the National Key R\&D Program of China (Grant No.~2019YFA0308200), the Innovation Program for Quantum Science and Technology (Grant No.~2023ZD0300200), the National Natural Science Foundation of China (Grants No.~11974213, No.~92065112, No.~12175003 and No.~12361161602), NSAF (Grant No.~U2330201), Shandong Provincial Natural Science Foundation (Grants No.~ZR2020JQ05 and No.~ZR2023LLZ005), Taishan Scholar of Shandong Province (Grant No.~tsqn202103013), Shenzhen Fundamental Research Program (Grant No.~JCYJ20220530141013029), the Higher Education Discipline Innovation Project (``111") (Grant No. B13029) and the High-Performance Computing Platform of Peking University.

\section*{Competing Interests}
The authors declare no Competing Financial or Non-Financial Interests.

\section*{Author Contributions}
G.~S., J. A. S., Q.~Z., D.~G., X.~M. and X.~Y. concreted the theory. H. L. conceived and designed the experiment. T.~Z., L.~L., X.-J.~P. and H.~L. carried out the experiment and analyzed data. All authors
contributed to writing the manuscript.

%\section*{References}
%\bibliographystyle{naturemag}
%\bibliography{refexp}

\appendix
\renewcommand{\figurename}{\textbf{Supplementary Fig}}
\renewcommand{\tablename}{\textbf{Supplementary Table}}

\onecolumngrid
\section*{Supplementary note 1: Derivation of the bounds}
Given a quantum state $\rho$ in a $d$-dimensional Hilbert space, our task is to bound the Von Neumann entropy of $\rho$ with a function of the state purity $\mathcal P(\rho):=\mathrm{Tr}(\rho^2).$
The spectral decomposition of the quantum state is $\rho = \sum_{i=1}^d\lambda_i\ket{\psi_i}\bra{\psi_i}$, where $\{\ket{\psi_i}\}$ forms an orthonormal basis of the $d$-dimensional Hilbert space. The variational problem is then formulated as
\begin{equation}\label{Eq:maxd}
\begin{aligned}
\max/\min S(\rho)&= -\sum_{i=1}^d\lambda_i\log(\lambda_i)\\
s.t.~~ &\sum_{i=1}^d\lambda_i^2 = \gamma\\
&\sum_{i=1}^d\lambda_i=1\\
&0\le \lambda_i \le1, \forall i,
\end{aligned}
\end{equation}
where $\mathcal P(\rho)=\mathrm{Tr} \rho^2$ is the purity of $\rho$.

Intuitively, the vector $\lambda$ that maximizes $S$ is the one that spread as uniformly as possible; while the vector $\lambda$ that minimizes $S$ is the one that has the minimal number of nonzero large values. In the following, we will analytically solve this problem and confirm this intuition.
\subsection{Maximization}
First, we focus on the maximization problem with $d=3$. Note that when $d=2$, the solution to the constraints of~\cref{Eq:maxd} is unique and the optimization problem will be trivial. Without loss of generality, we assume $\lambda_1\ge \lambda_2\ge \lambda_3$. Then the problem can be stated as
\begin{equation}\label{Eq:max3}
\begin{aligned}
\max S(\rho)&= -\lambda_1\log(\lambda_1)-\lambda_2\log(\lambda_2)-\lambda_3\log(\lambda_3)\\
s.t.~~ &\lambda_1^2+\lambda_2^2+\lambda_3^2 = \gamma\\
&\lambda_1+\lambda_2+
\lambda_3=1\\
&1\ge \lambda_1\ge \lambda_2\ge \lambda_3\ge0. \\
\end{aligned}
\end{equation}
We prove that the maximum is reached with the following Lemma.
\begin{lemma}\label{lamme:max3}
The solution to the maximization problem in~\cref{Eq:max3} is given by
\begin{equation}\label{lamme:max3solution}
\begin{aligned}
\lambda_1 &= \frac{1}{3}+\sqrt{\frac{2}{3}\left(\mathcal P(\rho)-\frac{1}{3}\right)},\\
  \lambda_2 & = \lambda_3 = \frac{1-\lambda_1}{2}.
\end{aligned}
\end{equation}
\end{lemma}

\begin{proof}
The differential of the entropy function $S(\rho)$ and the constraints are given by
\begin{equation}\label{Eq:ConstraintDiff1}
\begin{aligned}
 dS&= -\left(\frac{1}{\ln(2)}+\log \lambda_1\right)d\lambda_1-\left(\frac{1}{\ln(2)}+\log \lambda_2\right)d\lambda_2\\
 &-\left(\frac{1}{\ln(2)}
 +\log \lambda_3\right)d\lambda_3\\
 \end{aligned}
\end{equation}
and
\begin{equation}\label{Eq:ConstraintDiff2}
\begin{aligned}
&\lambda_1d\lambda_1+\lambda_2d\lambda_2+\lambda_3d\lambda_3 = 0\\
 &d\lambda_1+d\lambda_2+d\lambda_3=0,
  \end{aligned}
\end{equation}
respectively. We rewrite~\cref{Eq:ConstraintDiff2} to
\begin{equation}\label{Eq:}
\begin{aligned}
 d\lambda_1=&-\frac{(\lambda_3-\lambda_2)}{\lambda_1-\lambda_2}d\lambda_3,\nonumber\\
 d\lambda_2=&-\frac{(\lambda_1-\lambda_3)}{\lambda_1-\lambda_2}d\
 \lambda_3.
  \end{aligned}
\end{equation}
Thus, the differential of the entropy function becomes
\begin{equation}\label{Eq:Constraint}
\begin{aligned}
dS(\rho) &= \frac{d\lambda_3}{\lambda_1-\lambda_2}[(\lambda_3-\lambda_2)\log \lambda_1 +(\lambda_1-\lambda_3)\log \lambda_2 \\
&\quad +(\lambda_2-\lambda_1)\log \lambda_3] \\
&= (\lambda_2-\lambda_3)\left[-\frac{\log \lambda_1-\log \lambda_2}{\lambda_1-\lambda_2}+\frac{\log \lambda_3-\log \lambda_2}{\lambda_3-\lambda_2}\right]d\lambda_3
\end{aligned}
\end{equation}
Since the function $\log \lambda$ is concave for $\lambda\in[0,1]$,
for $\lambda_1\ge \lambda_2\ge \lambda_3$,
\begin{equation}
\frac{\log \lambda_1-\log \lambda_2}{\lambda_1-\lambda_2} \le \frac{\log \lambda_3-\log \lambda_2}{\lambda_3-\lambda_2}.\nonumber
\end{equation}
Thus, $dS(\rho)/d\lambda_3\ge 0$. To reach the maximum of $S(\rho)$, we thus only need to set $\lambda_3$ to be its maximum, which happens when $\lambda_2=\lambda_3$. Together with the constraints, then we can solve the equations and show that the solution to the maximization problem is given in~\cref{lamme:max3solution}.
\end{proof}
Now, we can solve the maximization problem of~\cref{Eq:maxd} for a general case of $d$.
\begin{theorem}\label{theorem:max}
Suppose $\lambda_1\ge \lambda_2 \ge \dots \lambda_d$. The solution to the maximization problem in~\cref{Eq:maxd} is
\begin{equation}\label{Eq:solutionmaxd}
\begin{aligned}
  \lambda_1 &= \frac{1}{d}+\sqrt{\frac{d-1}{d}\left(\mathcal P(\rho)-\frac{1}{d}\right)},\\
  \lambda_2 &=\lambda_3=\dots=\lambda_d = \frac{1-\lambda_1}{d-1}.
\end{aligned}
\end{equation}
\end{theorem}
\begin{proof}
The solution in~\cref{Eq:solutionmaxd} is exactly determined when setting $\lambda_2=\lambda_3=\dots=\lambda_d$. Suppose the maximization problem solution is not this one, then we must have that $\lambda_2>\lambda_d$. In the following, we prove the contradiction by showing that changing the values of $\lambda_1, \lambda_2, \lambda_d$ would make the entropy $S(\rho)$ larger while fixing all other values ($\lambda_3, \lambda_4, \dots, \lambda_{d-1}$) and the constraints. Now the constraints for $\lambda_1$, $\lambda_2$, and $\lambda_d$ becomes
\begin{equation}\label{Eq:cons12d}
\begin{aligned}
&\lambda_1^2+\lambda_2^2+\lambda_d^2 = a\\
 &\lambda_1+\lambda_2+\lambda_d=b.
 \end{aligned}
\end{equation}
By defining $\lambda_1' = \lambda_1/b$, $\lambda_2' = \lambda_2/b$, $\lambda_d' = \lambda_d/b$, the relations become
\begin{equation}\label{Eq:cons12dp}
\begin{aligned}
\lambda_1'^2+\lambda_2'^2+\lambda_d'^2 &= a/b^2\\
 \lambda_1'+\lambda_2'+\lambda_d'&=1.
 \end{aligned}
\end{equation}
The entropy function is
\begin{equation}\label{Eq:}
\begin{aligned}
S(\rho)&=-\sum_{i=1}^d \lambda_i\log(\lambda_i),\\
 &=S_{1,2,d}(\rho) +S_r(\rho),
 \end{aligned}
\end{equation}
where $S_{1,2,d}(\rho) = -\lambda_1\log(\lambda_1)-\lambda_2\log(\lambda_2)-\lambda_d\log(\lambda_d)$ and $S_r(\rho)=-\sum_{i=3}^{d-1} \lambda_i\log(\lambda_i)$. Since $S_r(\rho)$ is fixed, we need to maximize $S_{1,2,d}(\rho)$, which can also be represented as
\begin{equation}\label{Eq:}
\begin{aligned}
&S_{1,2,d}(\rho) = -b\lambda_1'\log(b\lambda_1')-b\lambda_2'\log(b\lambda_2')-b\lambda_d'\log(b\lambda_d')\\
&= b[-\lambda_1'\log(\lambda_1')-\lambda_2'\log(\lambda_2')-\lambda_d'\log(\lambda_d')]-b\log b\\
 \end{aligned}
\end{equation}
Denoting $S_{1,2,d}'(\rho) = -\lambda_1'\log(\lambda_1')-\lambda_2'\log(\lambda_2')-\lambda_d'\log(\lambda_d')$, this optimization problem has the same form of~\cref{Eq:max3}. Then Lemma~\ref{lamme:max3} indicates that the maximum of $S_{1,2,d}'(\rho)$ given the constraints in~\cref{Eq:cons12dp} is reached when $\lambda_2'=\lambda_d'$. In other words, the maximum of $S_{1,2,d}(\rho)$ given the constrains of~\cref{Eq:cons12d} is saturated with $\lambda_2=\lambda_d$, which contradicts $\lambda_2>\lambda_d$. Therefore, the solution to the maximization problem is given by~\cref{Eq:solutionmaxd}.
\end{proof}
\subsection{Minimization}
Now, we consider the solution to the minimization of~\cref{Eq:maxd}.  Similarly, we first consider the minimization with $d=3$ and $\lambda_1\ge \lambda_2\ge \lambda_3$,
\begin{equation}\label{Eq:min3}
\begin{aligned}
\min S(\rho)&= -\lambda_1\log(\lambda_1)-\lambda_2\log(\lambda_2)-\lambda_3\log(\lambda_3)\\
s.t.~~ &\lambda_1^2+\lambda_2^2+\lambda_3^2 = \mathcal P(\rho)\\
 &\lambda_1+\lambda_2+\lambda_3=1\\
&1\ge \lambda_1\ge \lambda_2\ge \lambda_3\ge 0. \\
 \end{aligned}
\end{equation}
\begin{lemma}\label{lamme:min3}
The solution to the minimization problem in~\cref{Eq:min3} is reached either when $\lambda_1=\lambda_2$ or $\lambda_3=0$.
\end{lemma}
\begin{proof}
From the proof of Lemma \ref{lamme:max3}, we already showed that $dS(\rho)/d\lambda_3\ge 0$. Therefore, the lower bound of $S(\rho)$ is reached when $\lambda_3$ takes its minimum. As $2(\lambda_1^2+\lambda_2^2)\ge (\lambda_1+\lambda_2)^2$, according to~\cref{Eq:min3}, we have
\begin{equation}\label{Eq:}
\begin{aligned}
2(\mathcal P(\rho)-\lambda_3^2)\ge (1-\lambda_3)^2.
 \end{aligned}
\end{equation}
The lower bound for $\lambda_3$ is
\begin{equation}
\begin{aligned}
\lambda_3\ge \max \left\{0, \frac{1-\sqrt{6\mathcal P(\rho) - 2}}{3}\right\}\nonumber
 \end{aligned}
\end{equation}
Thus, when $\mathcal P(\rho) \ge 1/2$, the minimal possible value for $\lambda_3$ is 0. When $1/3\le \mathcal P(\rho) < 1/2$, the minimal possible value for $\lambda_3$ is $\frac{1-\sqrt{6\mathcal P(\rho) - 2}}{3}$ and $\lambda_1=\lambda_2=(1-\lambda_3)/2$. Note that $\mathcal P(\rho)\ge1/3$ for $d=3$.
\end{proof}

Now, we can show the general solution to the minimization of~\cref{Eq:maxd}.
\begin{theorem}\label{theorem:min}
Suppose $\lambda_1\ge \lambda_2 \ge \dots \lambda_k$, the solution to the minimization problem in~\cref{Eq:maxd} is
\begin{equation}\label{Eq:solutionmind}
\begin{aligned}
\lambda_1=\lambda_2=\dots=\lambda_{k-1}&=\frac{1-\alpha}{k-1},\\
\lambda_{k} &= \alpha,\\
\lambda_{k+1}=\dots=\lambda_d &= 0.
\end{aligned}
\end{equation}
Here,
\begin{equation}
  \alpha=\frac{1}{k} - \sqrt{ (1-1/k)(\mathcal P(\rho)-1/k)}
\end{equation}
and $k$ is the integer such that $ \frac{1}{k} \le  \mathcal P(\rho)\le \frac{1}{k-1}$.
\end{theorem}
\begin{proof}
Suppose we always have the solution in the form as
\begin{equation}\label{Eq:solutionPossible}
\begin{aligned}
  \lambda_1=\lambda_2=\dots=\lambda_{k-1}, \lambda_{k}, \lambda_{k+1}=\dots \lambda_d = 0.
\end{aligned}
\end{equation}
Otherwise, there must exist three $\lambda_i,\lambda_j,\lambda_k$ such that $\lambda_i> \lambda_j\ge \lambda_k$ and $\lambda_k\ne 0$. Following a similar argument in the proof of Theorem \ref{theorem:max}, we can show that this contradicts  Lemma \ref{lamme:min3}.

According to~\cref{Eq:solutionPossible}, we have
\begin{equation}\label{Eq:}
\begin{aligned}
  (k-1)\lambda_1^2 + \lambda_k^2 &= \mathcal P(\rho),\nonumber\\
  (k-1)\lambda_1+\lambda_{k}&=1,\\
  k&\le d
\end{aligned}
\end{equation}
We can show that the possible integer value for $k$ is unique. That is,
\begin{equation}\label{Eq:}
\begin{aligned}
 k[(k-1)\lambda_1^2 + \lambda_k^2] &\ge [  (k-1)\lambda_1+\lambda_{k}]^2\nonumber\\
 &\ge (k-1)[(k-1)\lambda_1^2 + \lambda_k^2]
 \end{aligned}
\end{equation}
Equivalently, we have
\begin{equation}\label{Eq:}
\begin{aligned}
 k\mathcal P(\rho) \ge 1  \ge (k-1)\mathcal P(\rho),\nonumber
 \end{aligned}
\end{equation}
hence
\begin{equation}\label{Eq:}
\begin{aligned}
\frac{1}{\mathcal P(\rho)}\le & \,\,k\,\,\le  \frac{1}{\mathcal P(\rho)} + 1,\nonumber\\
 \frac{1}{k} \le  & \,\, \mathcal P(\rho) \,\,\le \frac{1}{k-1}.
 \end{aligned}
\end{equation}
\end{proof}
\subsection{Upper and lower bounds to coherence and entanglement}
We now call $\{\lambda_{i,\rho}^{\text{M}}\}, \{\lambda_{i,\rho}^{\text{m}}\}$ the vectors solving the maximization and the minimization, respectively. Given a bipartite state $\rho_{AB}$, by minimizing (maximizing) the marginal purity on $B$ subsystem and maximizing (minimizing) the global purity, one has \\
\emph{Result 1~(Eq.~(2) of the main text)}--- Given a quantum state $\rho_{AB}\in \mathcal{H}_{d_A}\otimes \mathcal{H}_{d_B}$, and defining $\rho_B = \mathrm{Tr}_A\rho_{AB}$, its coherent information $I(A\rangle B)$ is bounded as follows: 
  \begin{eqnarray}  \label{thebound2}
 l_e(\rho_{AB})\leq I(A\rangle B)\leq u_e(\rho_{AB}),
 \end{eqnarray} 
{\footnotesize
 \begin{eqnarray} 
 &&l_e(\rho_{AB})= \nonumber\\
&-&(1-\lambda^{\text{m}}_{k_{\rho_B},\rho_B}) \log \lambda^{\text{m}}_{1,\rho_B} 
 -\lambda^{\text{m}}_{k_{\rho_B},\rho_B} \log  \lambda^{\text{m}}_{k_{\rho_B},\rho_B}\nonumber\\
  &+& (1-\lambda^{\text{M}}_{1,\rho_{AB}}) \log \frac{(1-\lambda^{\text{M}}_{1,\rho_{AB}})}{(d-1)} + \lambda^{\text{M}}_{1,\rho_{AB}} \log {\lambda^{\text{M}}_{1,\rho_{AB}}},\nonumber\\
&&u_c(\rho_{AB})=\nonumber\\
&& (1-\lambda^{\text{m}}_{k_{\rho_{AB}},\rho_{AB}}) \log \lambda^{\text{m}}_{1,\rho_{AB}} +\lambda^{\text{m}}_{k_{\rho_{AB}},\rho_{AB}} \log  \lambda^{\text{m}}_{k_{\rho_{AB}},\rho_{AB}}\nonumber\\
 &-&  (1-\lambda^{\text{M}}_{1,\rho_B})\log \frac{(1-\lambda^{\text{M}}_{1,\rho_B})}{(d_B-1)} - \lambda^{\text{M}}_{1,\rho_B} \log \lambda^{\text{M}}_{1,\rho_B}\nonumber.
\end{eqnarray}
}

By minimizing (maximizing) the coherence of the dephased state $\rho_d=\sum_i\ket{i}\!\!\bra{i}\rho\ket{i}\!\!\bra{i}$, and maximizing (minimizing) the coherence of the state under study, we obtain lower (upper) bounds to the relative entropy of coherence:\\
\emph{Result 2~(Eq.~(5) of the main text)} --- The relative entropy of coherence $C_{\mathrm{RE}}(\rho)$ is  bounded as follows:
\begin{eqnarray}\label{Eq:Estimationlb2}
l_c(\rho)\leq C_{\mathrm{RE}}(\rho)\leq u_c(\rho), 
\end{eqnarray}
{\footnotesize
\begin{eqnarray}
 &&l_c(\rho)= -(1-\lambda^{\text{m}}_{k_{\rho_d},\rho_d}) \log \lambda^{\text{m}}_{1,\rho_d}-\lambda^{\text{m}}_{k_{\rho_d},\rho_d} \log \lambda^{\text{m}}_{k_{\rho_d},\rho_d}\nonumber\\
  &+& (1-\lambda^{\text{M}}_{1,\rho})\log\frac{(1-\lambda^{\text{M}}_{1,\rho})}{(d-1)} + \lambda^{\text{M}}_{1,\rho} \log{\lambda^{\text{M}}_{1,\rho}},\nonumber\\
 &&u_c(\rho)=  (1-\lambda^{\text{m}}_{k_{\rho},\rho}) \log \lambda^{\text{m}}_{1,\rho} +\lambda^{\text{m}}_{k_{\rho},\rho} \log \lambda^{\text{m}}_{k_{\rho},\rho}\nonumber\\
 &-&  (1-\lambda^{\text{M}}_{1,\rho_d}) \log \frac{(1-\lambda^{\text{M}}_{1,\rho_d})}{(d-1)} - \lambda^{\text{M}}_{1,\rho_d} \log \lambda^{\text{M}}_{1,\rho_d}\nonumber.
\end{eqnarray}
}

\subsection{Tightness of the bounds}
\begin{figure*}[t!]
\centering
\includegraphics[width=0.7\linewidth]{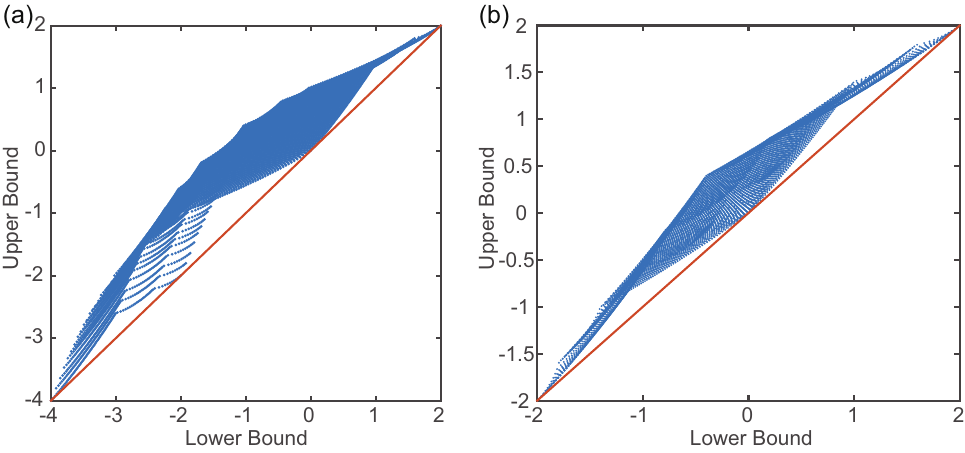}
\caption{Numerical simulations of lower and upper bounds of (a) coherent information and (b) relative entropy of coherence.}\label{fig:simulation}
\end{figure*}
To investigate the tightness of our bounds, we simulate the lower and upper bounds of  coherent information on four-qubit state $\rho_{AB}$ with $|A|=|B|=2$, and the results are shown in Supplementary Fig.~\ref{fig:simulation}(a). Each point displayed corresponds to the bounds at some pair of purity values $\mathcal{P}(\rho_{AB})$ and $\mathcal{P}(\rho_{B})$. The results are shown in Supplementary Fig.~\ref{fig:simulation}(a). Also, we simulate the lower and upper bounds of relative entropy of coherence on four-dimensional state~($d=4$), and the results are shown in Supplementary Fig.~\ref{fig:simulation}(b). Each point displayed corresponds to the bounds at some pair of purity values $\mathcal{P}(\rho)$ and $\mathcal{P}(\rho_{d})$. We can see that overall the upper and lower bounds give a good estimation of the the coherent information and the relative entropy of coherence. These bounds are tight when their values are close to the minimal~(maximally mixed states) and maximal~(maximally entangled states and maximally coherent states).

We use $\epsilon$ to characterize the ``tightness" of lower and upper bounds for \emph{pure states $\rho$}. For coherent information, the bounds are tight~($u_e(\rho)=l_e(\rho)$) for pure states~($\mathcal{P}(\rho_{AB}=1)$) with $\mathcal{P}(\rho_B)=\frac{1}{d_B}$, where $B$ is the subsystem of $\rho=\rho_{AB}$ and $d_B$ is the dimension of system $B$. The difference $\epsilon_e=\mathcal{P}(\rho_B)-1/d_B$ is related to the tightness of $u_e(\rho)$ and $l_e(\rho)$. For example, we consider a 4-qubit GHZ state $\ket{\text{GHZ}_4}=\frac{1}{\sqrt{2}}(\ket{0000}+\ket{1111})_{1234}$ with bi-partition of $A=\{1, 2, 3\}$ and $B=\{4\}$. The purity of subsystem $B$ is $\mathcal{P}(\rho_B)=\frac{1}{2}$ indicating the maximal entanglement of $\ket{\text{GHZ}_4}$.

For relative entropy of coherence, the bounds are tight~($u_c(\rho)=l_c(\rho)$) for pure state~($\mathcal{P}(\rho)=1$) with diagonal matrix of $\rho_d=\frac{1}{d}\mathbb I_d$~($\mathcal{P}(\rho_d)=\frac{1}{d}$). The difference $\epsilon_c=\mathcal{P}(\rho_d)-1/d$ is related to the tightness of $u_c(\rho)$ and $l_c(\rho)$. For example, we consider a single-qubit state $\ket{+}=\frac{1}{\sqrt{2}}(\ket{0}+\ket{1})$, the purity of its diagonal matrix is $\mathcal{P}(\rho_d)=\frac{1}{2}$ indicating $\ket{+}$ is a maximally coherent state. 

In theory, $\epsilon_e$ and $\epsilon_c$ provide mathematical criteria for the tightness of our method. $\epsilon_e(\epsilon_c)\to0$ indicates the maximally entangled state~(maximally coherent state), which is of particular interest in quantum information science.
%\onecolumngrid
\section*{Supplementary Note 2: Details of experimental realizations and results}
Before discussing the experimental details, we provide an overview of the crucial optical components and their respective functions employed in our experiment:
\begin{itemize}
    \item Waveplates: We utilize both half-wave plates (HWPs) and quarter-wave plates (QWPs) to perform unitary operations ($U$) on photons. The angle between the fast axis of the waveplate and the vertical polarization direction is denoted as $\omega$ for HWPs and $\vartheta$ for QWPs. These waveplates induce specific unitary transformations on the quantum state of photons, enabling the desired experimental operations. The unitary transformations of HWPs and QWPs can be expressed as:
\begin{align}
    U_{\text{HWP}}(\omega) &= -\begin{pmatrix}
 \cos 2\omega &  \sin 2\omega \\
 \sin 2\omega &  -\cos 2\omega \\
\end{pmatrix}, \\
U_{\text{QWP}}(\vartheta) &= \frac{1}{\sqrt{2}}\begin{pmatrix}
 1+i\cos 2\vartheta &  i\sin 2\vartheta \\
 i\sin 2\vartheta &  1-i\cos 2\vartheta \\
\end{pmatrix}
\end{align}

\item Polarization beam splitter~(PBS): This component separates photons with perpendicular polarization directions. It transmits horizontally polarized photons while reflecting vertically polarized photons.

\item Beam displacer~(BD): The BD allows vertically polarized photons to pass through in the original direction, while horizontally polarized photons are displaced. In our experiment, we use a BD with a length of 28.3 mm, resulting in a displacement offset of 3~mm.
\end{itemize}

\subsection*{A. Preparation of $\ket{\text{GHZ}_\theta}$ and $\ket{\psi_{2, \theta}}$}
To generate $\ket{\text{GHZ}_\theta}$, photon 1 and 2 pass through two BDs after calibration of their polarization as shown in~\cref{Fig:state}(a). Calling $h(v)$ the deviated (transmitted) spatial modes, the step-by-step description of optical elements to generate $\ket{\text{GHZ}_\theta}$ is 
\begin{equation}\label{GHZ}
\begin{aligned}
&\cos\theta\ket{HH}_{12}+\sin\theta\ket{VV}_{12}\\
\xrightarrow{\rm{BD1}}&\cos\theta\ket{HhH}_{11^\prime2}+\sin\theta\ket{VvV}_{11^\prime2}\\
\xrightarrow{\rm{BD2}}&\cos\theta\ket{HhHh}_{11^\prime22^\prime}+\sin\theta\ket{VvVv}_{11^\prime22^\prime}.
\end{aligned}
\end{equation}
\begin{figure}[h]
\centering
\includegraphics[width=0.5\linewidth]{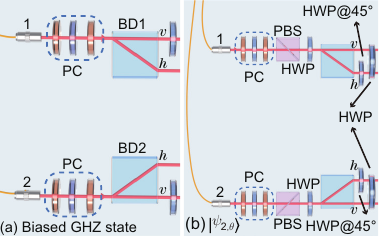}
\caption{(a) The setup to generate biased four-qubit GHZ states $\ket{\text{GHZ}_\theta}$. (b) The setup to generate two-copy states $\ket{\psi_{2, \theta}}$. PC: polarization control is a combination of QWP, HWP and QWP which can be set to calibrate the variation of polarization caused by the twist of fibre.}\label{Fig:state}
\end{figure}
In the experiment, by changing the value of $\theta$, we totally generate eleven $\rho_{\text{GHZ}_{\theta}}$. For each $\rho_{\text{GHZ}_{\theta}}$, we reconstruct the corresponding density matrix $\rho_{\text{GHZ}_{\theta}}$ using standard quantum state tomography (SQT). We calculate the fidelity $F=\tr(\rho_{\text{GHZ}_{\theta}}\ket{\text{GHZ}_{\theta}}\bra{\text{GHZ}_{\theta}})$. The results are summarized in~Supplementary Table~\ref{Tab:fidelity}.
\begin{table}[h]
\centering
\caption{The fidelities of the prepared $\rho_{\text{GHZ}_{\theta}}$.}\label{Tab:fidelity}
\begin{tabular}{|c|c|}
\hline
$\theta$ & $F$\\
\hline
0 & $0.9293\pm0.0004$\\
\hline
$\pi/20$ & $0.9050\pm0.0006$\\
\hline
$2\pi/20$ & $0.9057\pm0.0007$\\
\hline
$3\pi/20$ &  $0.9331\pm0.0005$ \\
\hline
$4\pi/20$ & $0.9136\pm0.0006$\\
\hline
$5\pi/20$ &  $0.9124\pm0.0007$\\
\hline
$6\pi/20$ & $0.9155 \pm0.0003$ \\
\hline
$7\pi/20$ & $0.8850 \pm0.0007$ \\
\hline
$8\pi/20$ & $0.8833 \pm0.0009$\\
\hline
$9\pi/20$ & $0.9155 \pm0.0005$\\
\hline
$10\pi/20$ & $0.9290 \pm0.0004$\\
\hline
\end{tabular}
\end{table}

The setup to prepare $\ket{\psi_{2, \theta}}\otimes\ket{\psi_{2, \theta}}$ is illustrated in~Supplementary Fig.~\ref{Fig:state}(b). Photon 1 and 2 pass through two PBSs creating the resulting state of $\ket{HH}_{12}$. The step-by-step description of optical elements to generate $\ket{\psi_{2, \theta}}\otimes\ket{\psi_{2, \theta}}$ is 
%\begin{widetext}
    \begin{equation}\label{Eq:two copy1}
     \begin{aligned}
      \ket{HH}_{12}&\xrightarrow[\text{are applied on photon 1 and 2}]{\text{Two HWPs}@\theta/2}(\cos\theta\ket{H}_1+\sin\theta\ket{V}_1)\otimes(\cos\theta\ket{H}_2+\sin\theta\ket{V}_2)\\
& \xrightarrow[\text{are applied on photon 1 and 2}]{\rm{BD1}\&\rm{BD2}}(\cos\theta\ket{Hh}_{11^\prime}+\sin\theta\ket{Vv}_{11^\prime})\otimes(\cos\theta\ket{Hh}_{22^\prime}+\sin\theta\ket{Vv}_{22^\prime})\\
& \xrightarrow[\text{are applied on path}\,\, h\,\,\text{of photon 1 and 2}]{\text{Two HWPs}@45\degree}\ket{VV}_{12}\otimes[(\cos\theta\ket{h}_{1^\prime}+\sin\theta\ket{v}_{1^\prime})\otimes(\cos\theta\ket{h}_{2^\prime}+\sin\theta\ket{v}_{2^\prime})]=\ket{VV}_{12}\otimes\ket{\psi_{2,\theta}}_{1^\prime2^\prime}\\
& \xrightarrow[\text{are applied on both path of photon 1 and 2}]{\text{Two HWPs}@\theta/2}\ket{\psi_{2,\theta}}_{12}\otimes\ket{\psi_{2,\theta}}_{1^\prime2^\prime}
\end{aligned}
\end{equation}
%\end{widetext}
As shown in~\cref{Eq:two copy1}, the first $\ket{\psi_{2, \theta}}$ is encoded into the polarization degree of freedom (DOF) of both photons, while the second $\ket{\psi_{2, \theta}}$ is encoded on the path DOF of both photons.

\subsection*{B. Experimental setups to perform shadow estimation and collective measurements}
\begin{figure}[h]
\centering
\includegraphics[width=0.5\linewidth]{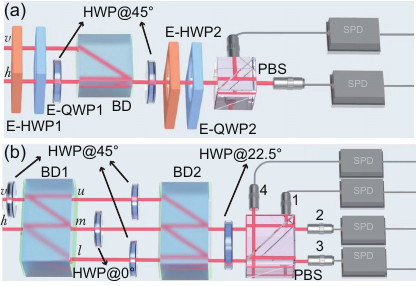}
\caption{Illustration of setups to perform (a) shadow estimation and (b) collective measurements.}\label{Fig:measurement}
\end{figure}

In shadow estimation, the single-qubit Clifford unitary operation $U_n\in \text{Cl}_2$ followed by a projective measurement on Pauli-Z basis is equivalent to measuring qubit on basis $J_n=U_n^{\dagger}ZU_n\in\left\{\pm X,\pm Y,\pm Z \right\}$. The experimental setup to perform measurement on arbitrary basis $(\alpha \ket{H}+\delta \ket{V})\otimes(\gamma\ket{h}+\mu\ket{v})$ is shown in~Supplementary Fig.~\ref{Fig:measurement}(a), where $\alpha \ket{H}+\delta \ket{V}$ is encoded on polarization DOF and $\gamma\ket{h}+\mu\ket{v}$ is encoded on path DOF. Projection on $\alpha \ket{H}+\delta \ket{V}$ is realized by combination of E-HWP1~(set at angle of $\theta_1$) and E-QWP1~(set at angle of $\upsilon_1$), while the projection of $\gamma\ket{h}+\mu\ket{v}$ and its orthogonal basis is realized by combination of E-HWP2~(set at angle of $\theta_2$), E-QWP2~(set at angle of $\upsilon_2$) and a PBS. A step-by-step description of implementation of such projective measurement is described by
%\begin{widetext} 
\begin{equation}
\begin{aligned}
(\alpha\ket{H}+\beta \ket{V})\otimes(\gamma\ket{h}+\delta\ket{v})
& \xrightarrow[\text{on both path}]{\text{E-HWP1}@\omega_1\&\text{E-QWP1}@\vartheta_1} \ket{H}\otimes(\gamma\ket{h}+\delta\ket{v}) \\
& \xrightarrow[\rm{on\,\,path\,\,}h]{\text{HWP}@45^\circ} \gamma\ket{Vh}+\delta\ket{Hv} \\
& \xrightarrow[\text{combine the path}\,\,h\,\,\text{and}\,\,v]{\rm{BD3}} \gamma\ket{V}+\delta\ket{H} \\
& \xrightarrow{\rm{HWP@45^\circ}} \gamma\ket{H}+\delta\ket{V} \\
& \xrightarrow[]{\text{E-HWP2}@\omega_2\&\text{E-QWP2}@\vartheta_2} \ket{H} \\ 
& \xrightarrow{\rm{PBS}}
\left\{  
\begin{array}{rcl}  
\text{transmitted}, &\text{results of}\,\,(\alpha\ket{H}+\beta \ket{V})\otimes(\gamma\ket{h}+\delta\ket{v}) \\
\text{reflected},  &\text{results of}\,\,(\alpha\ket{H}+\beta \ket{V})\otimes(\gamma\ket{h}+\delta\ket{v})^\perp 
\end{array},  
\right.
\end{aligned}
\end{equation}
where $(\gamma\ket{h}+\delta\ket{v})^\perp$ is the orthogonal state of $\gamma\ket{h}+\delta\ket{v}$. 
%\end{widetext}

The experimental setup to implement BSM between polarization-encoded qubit and path-encoded qubit is shown in~\cref{Fig:measurement}(b). Note that the projection on four Bell states is demonstrated with single experimental setting. A detailed description of~\cref{Fig:measurement}(b) is 
%\begin{widetext}
\begin{equation}
\begin{split}
&    \begin{cases}
    &\ket{\Psi^{+}} = \frac{1}{\sqrt{2}}(\ket{Hv}+\ket{Vh})\\
    &\ket{\Psi^{-}} = \frac{1}{\sqrt{2}}(\ket{Hv}-\ket{Vh})\\
    &\ket{\Phi^{+}} = \frac{1}{\sqrt{2}}(\ket{Hh}+\ket{Vv})\\
    &\ket{\Phi^{-}} = \frac{1}{\sqrt{2}}(\ket{Hh}-\ket{Vv})\\
\end{cases}
\xrightarrow[{\text{on path}\,\,v}]{\text{HWP}@45^\circ} 
\begin{cases}
    &\frac{1}{\sqrt{2}}(\ket{Vv}+\ket{Vh})\\
    &\frac{1}{\sqrt{2}}(\ket{Vv}-\ket{Vh})\\
    &\frac{1}{\sqrt{2}}(\ket{Hh}+\ket{Hv})\\
    &\frac{1}{\sqrt{2}}(\ket{Hh}-\ket{Hv})\\
\end{cases}
\xrightarrow[\text{Re-encoding}]{\text{BD1}} 
\begin{cases}
    &\frac{1}{\sqrt{2}}(\ket{Vu}+\ket{Vm})\\
    &\frac{1}{\sqrt{2}}(\ket{Vu}-\ket{Vm})\\
    &\frac{1}{\sqrt{2}}(\ket{Hl}+\ket{Hm})\\
    &\frac{1}{\sqrt{2}}(\ket{Hl}-\ket{Hm})\\
\end{cases}\\
&\xrightarrow[\text{HWP}@0^\circ\text{on path}\,\,m]{\text{HWPs}@45^\circ\text{on path}\,\,u\,\,\text{and}\,\,l} 
\begin{cases}
    &\frac{1}{\sqrt{2}}(\ket{Hu}-\ket{Vm})\\
    &\frac{1}{\sqrt{2}}(\ket{Hu}+\ket{Vm})\\
    &\frac{1}{\sqrt{2}}(\ket{Vl}+\ket{Hm})\\
    &\frac{1}{\sqrt{2}}(\ket{Vl}-\ket{Hm})\\
\end{cases}
\xrightarrow[\text{combine path}]{\text{BD2}} 
\begin{cases}
    &\ket{-}\ket{m}\\
    &\ket{+}\ket{m}\\
    &\ket{+}\ket{l}\\
    &\ket{-}\ket{l}\\
\end{cases}
\xrightarrow[\text{on path}\,\,m\,\,\text{and}\,\,l]{\text{HWP}@22.5^\circ} 
\begin{cases}
    &\ket{V}\ket{m}\\
    &\ket{H}\ket{m}\\
    &\ket{H}\ket{l}\\
    &\ket{V}\ket{l}\\
\end{cases}
\xrightarrow[]{\text{PBS}} 
\begin{cases}
    &\text{clicks on detector 1}\\
    &\text{clicks on detector 2}\\
    &\text{clicks on detector 3}\\
    &\text{clicks on detector 4}\\
\end{cases}.
\end{split}
\end{equation}
%\end{widetext}

\subsection*{C. Robustness of coherent information and relative entropy of coherence}
\begin{figure*}[t]
    \centering
    \includegraphics{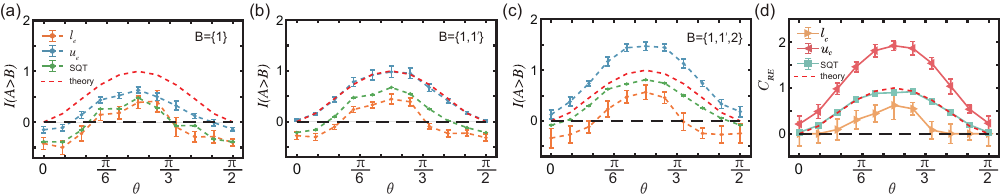}
    \caption{Experimental results of lower and upper bounds of $I(A\rangle B)$ and $C_{\mathrm{RE}}$ for the prepared $\rho_{\text{GHZ}_{\theta}}$. (a)-(c), the lower bound $l_e$~(orange dots) and upper bound $u_e$~(blue dots) with $B=\{1\}, \{1, 1^\prime\}$ and $\{1, 1^\prime, 2\}$ respectively. The green dots represent $I(A\rangle B)$ calculated with reconstructed $\rho_{\text{GHZ}_{\theta}}$, and the dashed red lines represents the theoretical predictions of ideal $\ket{\text{GHZ}_\theta}$. (d), the the lower bound $l_c$~(orange dots) and upper bound $u_c$~(red dots) of $C_\text{RE}(\rho_{\psi_{2,\theta}})$. The cyan dots represent $C_\text{RE}(\rho_{\psi_{2,\theta}})$ calculated with reconstructed $\rho_{\text{GHZ}_{\theta}}$ and the dashed red line represents theoretical prediction of ideal $\ket{\text{GHZ}_\theta}$.   }
    \label{fig:ideal}
\end{figure*}
In our experiment, the fidelity of prepared states are 0.9\% with respect to the target state, so that the results with standard state tomography~(QST) are treated as ``true values". This is because the coherent information and relative entropy of coherence are sensitive to noise. To clarify this, we present the theoretical predictions of $I(A\rangle B)$ with $B=\{1\}, \{1, 1^\prime\}$ and $\{1, 1^\prime, 2\}$~(red dashed lines) in Supplementary Fig.~\ref{fig:ideal}(a)-(c) and that of $C_\text{RE}$ in Supplementary  Fig.~\ref{fig:ideal}(d). From Supplementary Fig.~\ref{fig:ideal}, we observe that 
\begin{itemize}
\item[1.] There are always gaps between the theoretical predictions and the results from SQT for $I(A\rangle B)$, and the smaller the $|B|$ is the larger the gap is. 
\item[2.] The theoretical predictions of $C_\text{RE}$ agrees well with that from SQT, except the state $\rho_{\text{GHZ}_\theta}$ with $\theta=\pi/4$. 
\end{itemize}
The discrepancy between theoretical predictions and experimental results are mainly caused by experimental noise, which we will discuss separately in the following.

For the coherent information $I(A\rangle B)$ and $C_\text{RE}$, we calculate the values of $I(A\rangle B)$ on four-qubit state 
\begin{equation}\label{Eq:noisyGHZ1}
\rho_\text{noise}=(1-p)\ket{\text{GHZ}_{\pi/4}}\bra{\text{GHZ}_{\pi/4}}+p\frac{\mathbb I}{16}. 
\end{equation}
The results of $I(A\rangle B)$ with $p\in[0, 0.1]$ are shown in Supplementary Fig.~\ref{fig:simnoise}(a), where $I(A\rangle B)$ decreases faster with small $|B|$. In our experiment, the average fidelity of prepared $\ket{\text{GHZ}_\theta}$ is about 0.9 as listed in Supplementary Table~\ref{Tab:fidelity}, which introduces different discrepancies of $I(A\rangle B)$ with $|B|=1, 2$ and 3 as shown in Supplementary Fig.~\ref{fig:ideal}(a)-(c).
\begin{figure}[h!]
    \centering
    \includegraphics[width=\linewidth]{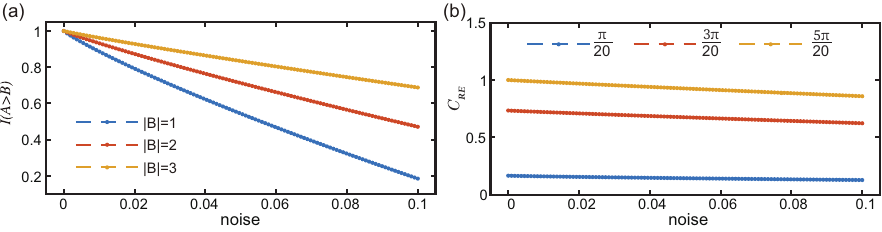}
    \caption{Simulated results of $I(A\rangle B)$ and $C_{\text{RE}}$ of noisy states. (a), $I(A\rangle B)$ of noisy state in Eq.~\ref{Eq:noisyGHZ1} with $|B|=1, 2$ and 3. (b), $C_{\text{RE}}$ of noisy state in Eq.~\ref{Eq:noisyGHZ2} with $\theta=\pi/20, 3\pi/20$ and $5\pi/20$.}
    \label{fig:simnoise}
\end{figure}
For the relative entropy of coherence $C_\text{RE}$, we calculate the value of $C_\text{RE}$ on four-qubit state  
\begin{equation}\label{Eq:noisyGHZ2}
\rho_\text{noise}=(1-p)\ket{\text{GHZ}_{\theta}}\bra{\text{GHZ}_{\theta}}+p\frac{\mathbb I}{16}. 
\end{equation}
with $\theta=\pi/20, 3\pi/20$ and $5\pi/20$ respectively. As shown in Supplementary Fig.~\ref{fig:simnoise}(b), $C_\text{RE}$ decreases faster with larger $\theta$ as $p$ increases, which explains the biggest discrepancy at $\theta=\pi/4$ in Supplementary  Fig.~\ref{fig:ideal}(d).

\end{document}